\documentclass[11pt]{article}
\usepackage[T1]{fontenc}
\usepackage[utf8]{inputenc}

\usepackage[margin=1in]{geometry}
\usepackage{amsfonts,amsmath,amssymb,amsthm}
\usepackage{cite,microtype,graphicx,hyperref}
\usepackage[capitalize,nameinlink]{cleveref}

\newtheorem{theorem}{Theorem}
\newtheorem{lemma}[theorem]{Lemma}
\newtheorem{corollary}[theorem]{Corollary}
\newtheorem{observation}[theorem]{Observation}

\theoremstyle{definition}
\newtheorem{definition}[theorem]{Definition}

\newcommand{\transpose}{^{\mathsf{T}}}
\newcommand{\symmetrize}{\widehat}
\newcommand{\dehn}{\mathcal{D}}
\newcommand{\symdehn}{\symmetrize{\dehn}}
\newcommand{\rank}{\operatorname{rank}}

\title{Orthogonal Dissection into Few Rectangles}
\author{David Eppstein\thanks{Department of Computer Science, University of California, Irvine. Donald Bren Hall, Irvine, CA 92697, USA. eppstein@uci.edu.}}

\date{ }

\begin{document}
\thispagestyle{empty}
\maketitle  

\begin{abstract}
We describe a polynomial time algorithm that takes as input a polygon with axis-parallel sides but irrational vertex coordinates, and outputs a set of as few rectangles as possible into which it can be dissected by axis-parallel cuts and translations. The number of rectangles is the rank of the Dehn invariant of the polygon. The same method can also be used to dissect an axis-parallel polygon into a simple polygon with the minimum possible number of edges. When rotations or reflections are allowed, we can approximate the minimum number of rectangles to within a factor of two.
\end{abstract}

\section{Introduction}

Problems of rearranging polygonal shapes into simpler shapes, such as orthogonal polygons into rectangles, have many applications in such varied topics as VLSI design, DNA microarray layout, image processing, radiation therapy planning, and robot self-assembly~\cite{Epp-WG-09,Pat-CAD-77,HanHubLip-CT-02,CheIyeKas-TSE-88,FerSanSkl-CVGIP-84,Eng-DAM-09,Kal-EJC-09,LiZha-IROS-05}. One way to do this, but not the only way, is by subdivision. Slicing an orthogonal polygon horizontally through each vertex partitions it into rectangles, but may use more rectangles than necessary. Instead, an algorithm based on bipartite matching can find a partition into a minimum number of rectangles in polynomial time, even for polygons with holes. The algorithm finds axis-parallel segments through pairs of non-convex vertices, constructs a bipartite intersection graph of these segments, and uses the fact that in bipartite graphs, maximum independent sets can be found using maximum matchings. Slicing along a maximum independent set of segments, with additional slices through each non-convex vertex missed by the independent set, produces a set of as few rectangles as possible~\cite{FerSanSkl-CVGIP-84,LipLodLuc-FI-79,Oht-ISCAS-82,Epp-WG-09}.

In this work, we study an analogous problem of rectangle minimization for a different class of rearrangement methods, in which we allow sliced pieces to be rejoined. Slicing a polygon into pieces and rejoining them into another polygon is called \emph{dissection}. Potentially, dissection can produce many fewer rectangles,
but it is not obvious how to choose the dissection operations in such a way to produce as few rectangles as possible.  For example, the Greek cross of \cref{fig:greek-cross} requires three rectangles when partitioned, but has a three-piece dissection into one rectangle, as shown. In fact, every polygon (orthogonal or not) can be dissected into every other polygon of the same area; this is the \emph{Wallace–Bolyai–Gerwien theorem}~\cite{Bol-33,Ger-Crelle-33,Jac-AJM-12,WalLow-NSMR-14}. Therefore, a dissection into a single rectangle always exists. However, this dissection may rotate pieces and use non-axis-aligned cuts, both of which are unnatural for orthogonal polygons. Instead, we ask: if we consider dissections that use only axis-parallel slices, translations, and rejoining of the sliced pieces, without rotations, how few rectangles can we dissect a given shape into? For instance, the figure demonstrates that the answer for the Greek cross is one: it can be dissected into a single rectangle. We call this restricted class of dissections \emph{orthogonal dissections}. As defined, these do not allow $90^\circ$-rotations, but we will also consider an extended class of dissections in which rotation (or equivalently reflection across a diagonal reflection line) is allowed; we call these \emph{orthogonal dissections with rotation}.

\begin{figure}[t]
\centering
\includegraphics[width=\textwidth]{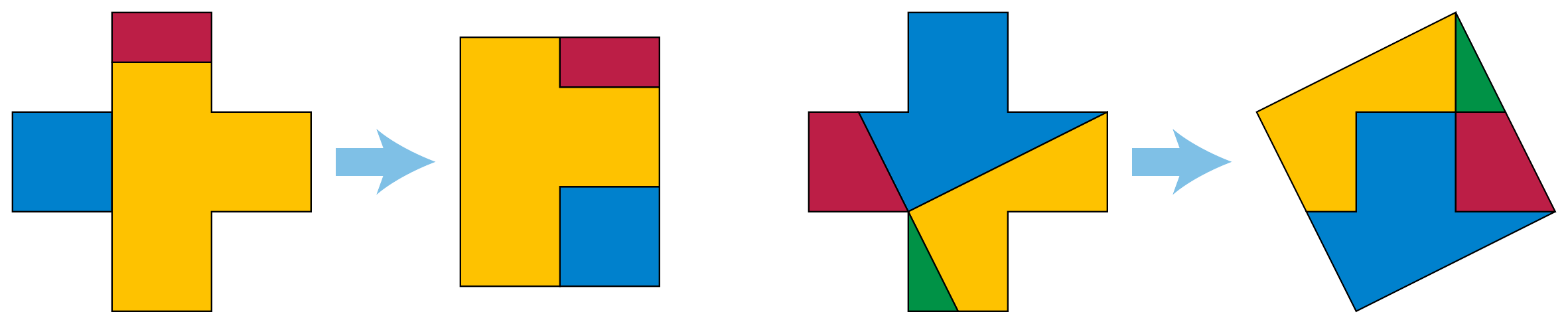}
\caption{Left: Dissection of a Greek cross into a rectangle, using only axis-parallel cuts and translation of pieces. Right: Dissection into a square using non-axis-parallel cuts~\cite{Fre-97}.}
\label{fig:greek-cross}
\end{figure}

Polyominoes (edge-to-edge unions of unit squares), such as the Greek cross of the figure, always have an orthogonal dissection into one rectangle. Simply subdivide a given polyomino into its constituent squares, let $n$ be the number of squares obtained in this way, and rearrange and join them into a $1\times n$ rectangle. To make the problem less trivial, we consider in this work polygons with irrational coordinates. As we will see, for these polygons, an orthogonal dissection into a single rectangle may not exist. To address the computational issues that this entails, we assume that all coordinates are presented as rational linear combinations of a \emph{rational basis}, a set of real numbers for which no nontrivial rational linear combination sums to zero.

A key technical tool that we use for analyzing dissections, following previous work along the same lines, is the \emph{Dehn invariant}. The Dehn invariant is a value living in an infinite-dimensional tensor space, usually used for three-dimensional polyhedral dissection problems. One polyhedron can be dissected into another if and only if they have the same volumes and Dehn invariants, and a polyhedron can be dissected to tile space if and only if its Dehn invariant is zero~\cite{Deh-MA-01,Dup-NTM-01,Jes-MS-68,Syd-CMH-65}. Another version of the Dehn invariant has also been used for orthogonal dissection of rectangles to rectangles, in order to prove that such a dissection exists if and only if the two rectangles have equal areas and rationally related sides~\cite{Ben-AMM-07,Deh-MA-03,Spa-AMM-04,Sti-98}. For instance, because the Greek cross of \cref{fig:greek-cross} (scaled to form a pentomino, with side length one) has an orthogonal dissection into a rectangle with dimensions $2\times 2\,\tfrac12$, it cannot also be orthogonally dissected into a $\sqrt 5\times \sqrt 5$ square. (Instead, it can be dissected into a square using only two straight but not axis-parallel cuts~\cite{Fre-97}.)

\subsection{New results}
Our main result is an algorithm that takes as input an orthogonal polygon with coordinates given in terms of a rational basis, that computes the minimum number of rectangles into which it can be dissected, and that constructs a family of rectangles of that minimum size into which it may be dissected (\cref{thm:minrect-algorithm}). The algorithm runs in time proportional to its input size (the number of points multiplied by the cardinality of the basis), in a model of computation in which rational arithmetic operations take constant time (as detailed in \cref{sec:model}). As we also show, this has strong implications for the possibility of dissecting a polygon into a prototile that can tile the plane: such a dissection exists if and only if the minimum number of rectangles is one or two (\cref{thm:rank-vs-tiling}).

To prove this, we extend the two-dimensional Dehn invariant from rectangles to orthogonal polygons more generally. We show that it is a complete invariant for orthogonal dissection: two orthogonal polygons have a dissection into each other if and only if they have the same Dehn invariant (\cref{thm:dissectability}). The Dehn invariant determines the area of a polygon, and (unlike the polyhedral Dehn invariant) whenever a value in the space of Dehn invariants has a positive area associated with it, it can be realized by an orthogonal polygon (\cref{lem:realizability}).
The key insight leading to our rectangle-minimization algorithm is that, as order-two tensors, Dehn invariants have significant structure beyond merely being equal or unequal to each other or zero. In particular, like matrices, they have a rank, and this rank is geometrically meaningful. We prove that, for the orthogonal Dehn invariant, the rank equals the minimum number of rectangles that can be obtained from an orthogonal dissection (\cref{thm:minrect-rank}). As we show, the same method can also be used to dissect a given axis-parallel polygon into a simple polygon with as few edges as possible (for polygons formed by orthogonal dissection). For a polygon whose Dehn invariant has rank $r$, and whose minimum number of rectangles is $r$, the minimum number of edges equals $2r+2$ (\cref{obs:edges}).

To extend these results to orthogonal dissections with rotation, we define a new symmetrized form of the Dehn invariant that is an invariant for this more general class of dissections. Its rank can differ from the rank of the Dehn invariant for dissections without rotation by an arbitrarily large factor. In this case, we do not have an efficient algorithm for the minimum number of rectangles into which we can dissect a given polygon, but we can describe a formula for it in terms of the minimum rank in a family of matrices having the same symmetrization. Additionally, as we show, the rank of the symmetrized Dehn invariant approximates the minimum number of rectangles to within a factor of two, providing a polynomial time approximation algorithm with approximation ratio two.

For the Dehn invariant of polyhedra, we do not have as precise a relation, but the rank of the Dehn invariant (if nonzero) provides a lower bound on the minimum number of edges in a polyhedron to which the given polyhedron can be dissected, and also on the minimum number of tetrahedra into which it can be dissected.

\section{Model of computation}
\label{sec:model}

The main objects of study in this work are the following:
\begin{definition}
We define an orthogonal polygon to be a bounded region of the plane whose boundary consists of finitely many axis-parallel line segments, allowing polygons with holes. We do not generally require this region to be connected or simply-connected; when we do, we call it a \emph{simple polygon}. A \emph{vertex} of a polygon is an endpoint of one of these line segments.
\end{definition}

We may represent these by specifying the coordinates of each vertex. Because the problems we consider are nontrivial only for polygons with irrational  coordinates, it is necessary to say something about how those coordinates are represented and how we compute with them.

\begin{definition}
We define a rational basis to be a system of finitely many real numbers that is linearly independent over $\mathbb{Q}$.
This means that, if a linear combination of basis elements with rational-number coefficients adds to zero, all coefficients must be zero.
\end{definition}

This is just the standard notion of a basis in linear algebra, applied to systems of real numbers that form vector spaces over the rational numbers. At most one member of a rational basis can be a rational number, because any two rational numbers $p$ and $q$ have a rational combination $\tfrac{1}{p}p-\tfrac{1}{q}q$ summing to zero. In general we allow either different bases for the $x$-coordinates and the $y$-coordinates (an $x$-basis and a $y$-basis) or a single combined basis; when we consider dissections with rotation, a combined basis will be more convenient. We require each vertex coordinate of a given orthogonal polygon to be a rational linear combination of basis elements, represented as a vector of rational-number coefficients, one for each basis element. The size of the input is the number of rational coefficients needed to describe all of the polygon vertices: the product of the number of vertices with the sum of the sizes of the $x$-basis and $y$-basis.

To compute the minimum number of rectangles in an orthogonal dissection, no additional information about the basis elements is necessary. Our algorithm for this version of the problem uses only rational-number arithmetic, and performs a polynomial number of arithmetic operations: essentially, only Gaussian elimination applied to a matrix whose coefficients are quadratic combinations of input coefficients. However, we need additional assumptions that allow computation with basis elements in order to verify that the input describes a polygon without edge crossings, or to construct the rectangles into which it can be dissected. To do these things, we need the following additional primitive operations:
\begin{itemize}
\item Find the sign of a rational combination of basis elements (that is, determine whether this value is positive, negative, or zero), or the sign of a rational combination of products of $x$-basis elements and $y$-basis elements.
\item Given any two rational combinations of basis elements, or of products of $x$-basis elements and $y$-basis elements, find a rational number between them.
\end{itemize}

We are not aware of past use of this specific computational model. However, exact computation using algebraic numbers is common in computational geometry implementation libraries~\cite{MehSch-SAMVM-01,KarFukvdH-ICMS-10}, and it is standard to represent such numbers as rational combinations of roots of a Galois polynomial, a special case of a rational basis. We have chosen the model described above in order to be able to state our results in a more general model that does not specify the algebraic nature of the numbers. In this way, the algorithms can apply as well to coordinates involving transcendental numbers such as $\pi$ and~$e$, as long as the primitive operations are available for these coordinates.

\section{The orthogonal Dehn invariant}
\label{sec:matrix}

The key tool for our results on orthogonal polygon dissection is the Dehn invariant, which we first define in a basis-specific way
as follows.

\begin{definition}
In terms of the given rational basis, the Dehn invariant $\dehn(P)$ of an orthogonal polygon $P$ can be described as a matrix of rational numbers,
with rows indexed by $y$-basis elements and columns indexed by $x$-basis elements, constructed as follows:
\begin{itemize}
\item Express the given polygon as a linear combination of rectangles $R_i$. For instance, if coordinate comparisons are available, we may slice the polygon horizontally through each non-convex vertex. If comparisons are unavailable, we may instead choose the line through one horizontal side as a base and consider the family of signed rectangles between each other horizontal side and this base line.
\item Express the width $w_i$ and height $h_i$ of each rectangle $R_i$ as a linear combination of basis elements with rational coefficients. The width is the difference of $x$-coordinates of right and left sides of the rectangle, the height is the difference of $y$-coordinates of top and bottom sides, and the difference of two linear combinations of basis elements produces another linear combination.
\item Construct a matrix $M_i$, the outer product of the expressions for $w_i$ and $h_i$. The coefficient of this matrix, for the column corresponding to an $x$-basis element $x_j$ and the row corresponding to a $y$-basis element $y_k$, is a rational number, the product of the coefficient of $x_j$ in $w_i$ and the coefficient of $y_k$ in $h_i$.
\item The Dehn invariant of the polygon is the sum of matrices $\sum_i M_i$.
\end{itemize}
\end{definition}

\begin{figure}[t]
\centering\includegraphics[width=0.8\columnwidth]{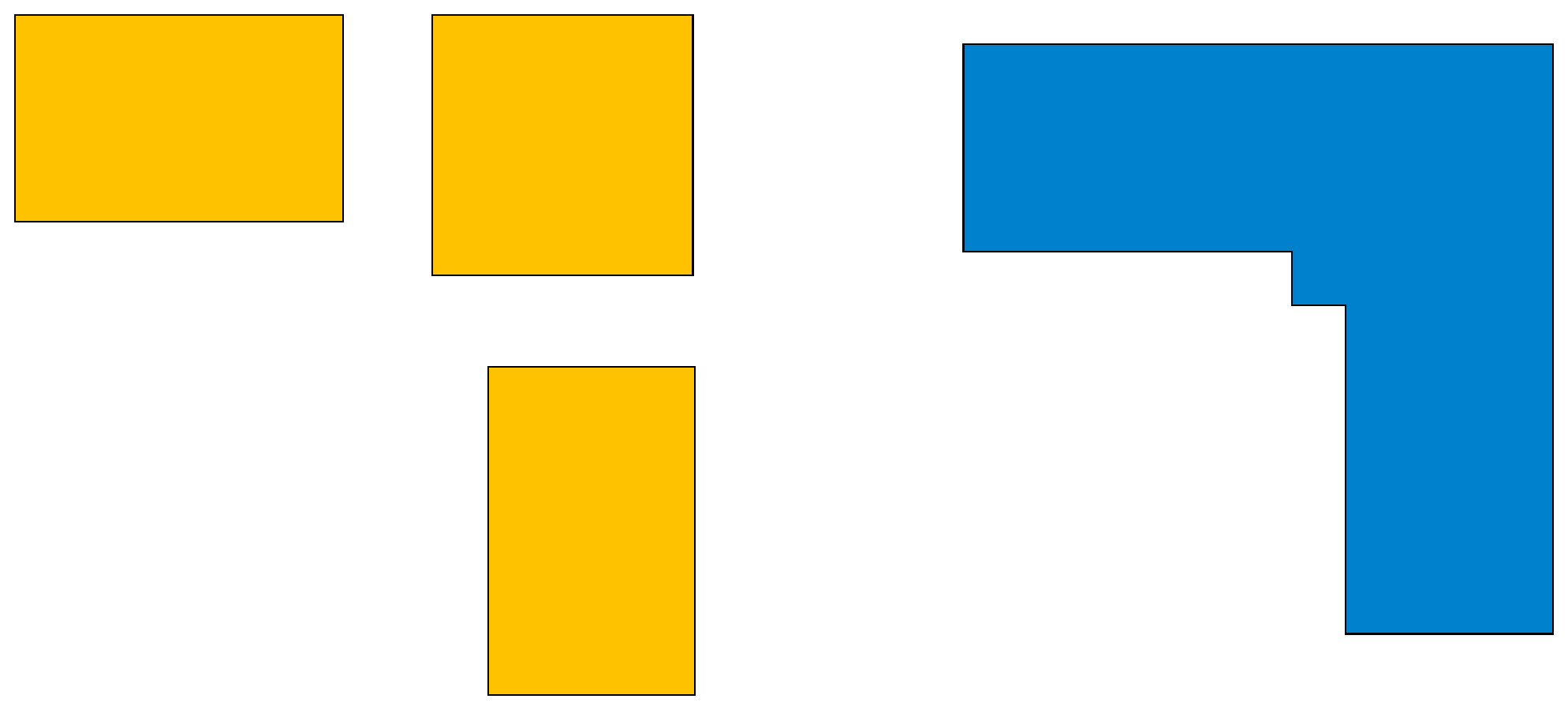}
\caption{Three rectangles with dimensions $2^{2/3}\times 1$, $2^{1/3}\times 2^{1/3}$, and $1\times 2^{2/3}$ (yellow), and a polygon formed by gluing them together (blue)}
\label{fig:trirect}
\end{figure}

For example, for the blue polygon in \cref{fig:trirect} and the rational basis $\{1,2^{1/3},2^{2/3}\}$, this definition would yield as the Dehn invariant the matrix
\[\begin{pmatrix} 0 & 0 & 1 \\ 0 & 1 & 0 \\ 1 & 0 & 0 \end{pmatrix},\]
as can be seen from its dissection into yellow rectangles in the figure. Each yellow rectangle has a width and a height that is one of the basis elements, so it contributes a single 1 coefficient to the total. The sum of the three 1 coefficients, from the three yellow rectangles, is the matrix above.

Instead of using a specific basis, one can describe the same thing in a basis-free way by writing that the Dehn invariant is an element of the tensor product of $\mathbb{Q}$-vector spaces $\mathbb{R}\otimes_{\mathbb{Q}}\mathbb{R}$, and can be determined as a sum of elements of this tensor product:\footnote{The Dehn invariant is often written as an element of a tensor product of abelian groups, rather than of vector spaces, using the notation $\mathbb{R}\otimes_{\mathbb{Z}}\mathbb{R}$ or, for the polyhedral invariant, $\mathbb{R}\otimes_{\mathbb{Z}}\mathbb{R}/\mathbb{Z}$. The group notation makes more sense for some contexts; for instance, it works for the polyhedral invariant in hyperbolic or spherical geometry, where linear scaling of polyhedra is not possible. But for our use of tensor rank, vector space notation is more convenient. For the equivalence of matrices and tensors see~\cite{Hac-12}.}
\[\dehn(P)=\sum_i h_i\otimes w_i.\]
It is an invariant of $P$, in the sense that its value (either thought of as a matrix for a specific basis or as a tensor) does not depend on the decomposition into rectangles used to compute it, and remains unchanged under orthogonal dissections; see \cref{sec:invariance}.

In contrast to the polyhedral Dehn invariant, the area of an orthogonal polygon $P$ can be recovered from its Dehn invariant under any basis, as the sum
\[\sum_j \sum_k \dehn(P)_{kj} x_j y_k\]
of products of matrix coefficients, $x$-basis elements, and $y$-basis elements. In this sense, it is meaningful to speak of the area of a Dehn invariant, rather than the area of a polygon.

\section{Invariance of the Dehn invariant}
\label{sec:invariance}

Previous works on the orthogonal Dehn invariant only appear to have considered it with regard to rectangles, rather than for orthogonal polygons more generally~\cite{Ben-AMM-07,Deh-MA-03,Spa-AMM-04,Sti-98}. Generalizing this past work, we prove here that it is an invariant of orthogonal polygons under dissection. Throughout this section, we use the abstract tensor space formulation of the orthogonal Dehn invariant; everything carries directly over to the formulation in any particular basis, according to standard principles on the invariance of linear algebra under different choices of basis.

\begin{lemma}
\label{lem:grid}
Let $R$ be a rectangle with height $h$ and width $w$. Suppose $R$ is subdivided arbitrarily by vertical and horizontal lines into a rectangular grid of smaller rectangles of heights $h_j$ and widths $w_k$, as depicted in \cref{fig:grid}. For all such subdivisions, $h\otimes w=\sum h_j\otimes w_k$.
\end{lemma}

\begin{figure}[t]
\centering\includegraphics[width=0.5\textwidth]{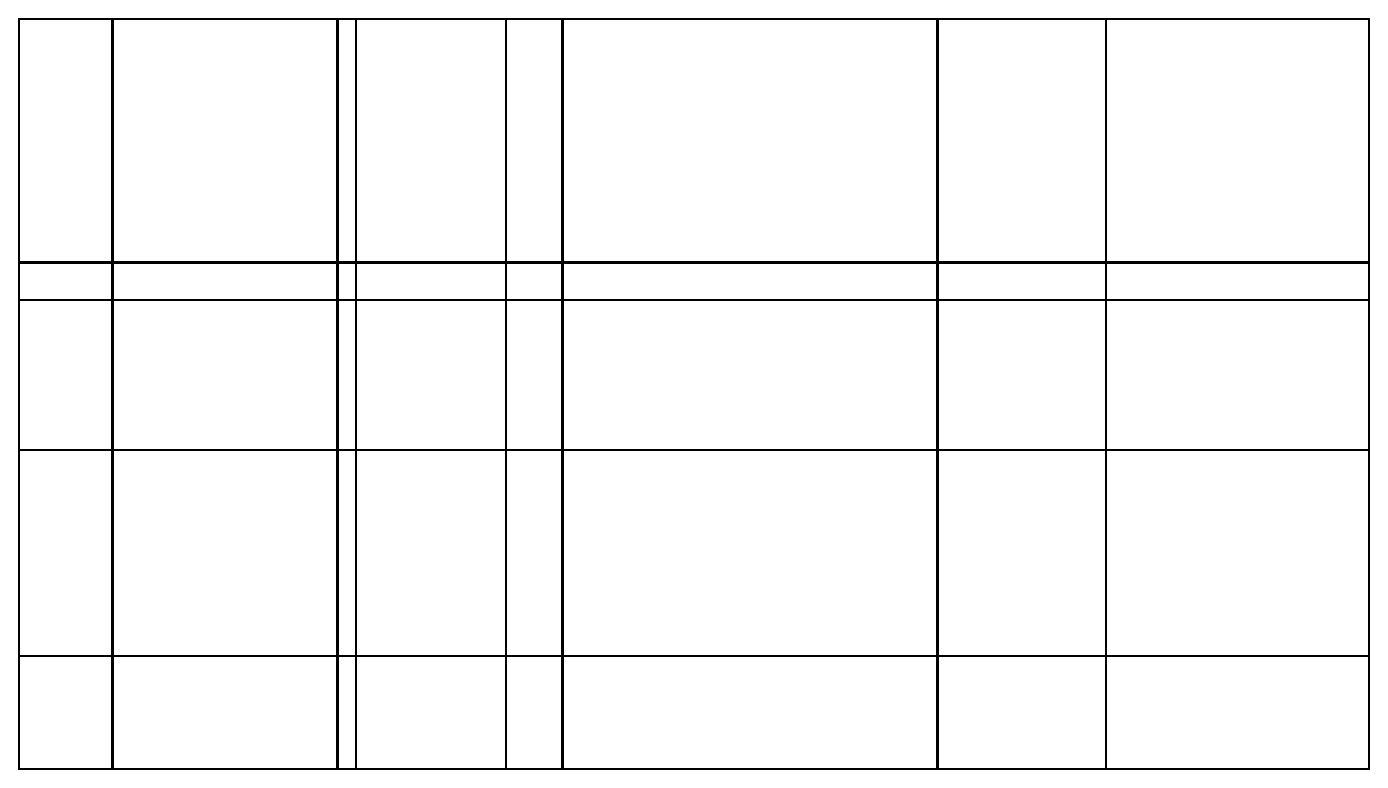}
\caption{Illustration for \cref{lem:grid}: subdividing a rectangle into a grid of smaller rectangles does not change its Dehn invariant.}
\label{fig:grid}
\end{figure}

\begin{proof}
This follows immediately from the facts that $\sum h_j=h$ and that $\sum w_k=w$, and from the bilinearity of tensors.
\end{proof}

\begin{lemma}
\label{lem:invariant-of-polygon}
Let $P$ be any orthogonal polygon. Then regardless of how $P$ is subdivided into rectangles $R_i$ of height $h_i$ and width $w_i$, the value $\sum h_i\otimes w_i$ will be unchanged. That is, $\dehn(P)=\sum h_i\otimes w_i$ is well-defined as an invariant of $P$.
\end{lemma}

\begin{proof}
Consider any two different subdivisions into rectangles $R_i$ and $R'_i$, and refine both subdivisions into a common subdivision by extending vertical and horizontal lines through all vertices of both $R_i$ and $R'_i$. By \cref{lem:grid}, this refinement does not change the sum over the rectangles in either subdivision. Because both of the sums coming from the initially given subdivisions are equal to the sum coming from their common refinement, they must be equal to each other.
\end{proof}

\begin{lemma}
\label{lem:invariant-of-dissection}
If two orthogonal polygons $P$ and $P'$ are related by an orthogonal dissection, then $\dehn(P)=\dehn(P')$. That is, the Dehn invariant remains invariant under orthogonal dissections.
\end{lemma}

\begin{proof}
We can refine any orthogonal dissection into a dissection for which all pieces are rectangles, and use those rectangles to calculate $\dehn(P)$ and $\dehn(P')$. Translating a rectangle obviously does not change its height or width, so the result follows from \cref{lem:invariant-of-polygon}.
\end{proof}

\section{The rank of the Dehn invariant}
\label{sec:rank}

Any tensor has a rank, the minimum number of terms needed to express it as a sum of tensor products. The Dehn invariants we are considering are order-two tensors over the field of rational numbers, and for any order-two tensor over any field, the rank of the tensor equals the rank of any matrix representing it for any basis over that field. As the rank of a matrix, it equals the minimum number of terms in an expression of the matrix as a sum of outer products of vectors~\cite{Hac-12}. Therefore, the rank of the Dehn invariant is just the rank of the matrix computed in \cref{sec:matrix}. It does not depend on the basis chosen to construct this matrix, and it can be computed using any standard algorithm for matrix rank, such as Gaussian elimination.

If an orthogonal polygon $P$ has an orthogonal dissection into $r$ rectangles with height $h_i$ and width $w_i$, we have seen that its Dehn invariant can be expressed as
\[\dehn(P)=\sum_{i=1}^r h_i\otimes w_i.\]
This is an expression as a sum of $r$ products, so the Dehn invariant has rank at most~$r$. Conversely, if an orthogonal polygon $P$ has a Dehn invariant with rank $r$, then it has an expression of exactly this form. However, not all terms of such an expression may be interpreted as describing rectangles. To come from a rectangle, a term $h_i\otimes w_i$ must have $h_i\cdot w_i>0$, in which case it can come from any rectangle of height $q\cdot |h_i|$ and width $|w_i|/q$ for any positive rational number $q$. All of these different rectangles produce the same value $h_i\otimes w_i$. But if the product $h_i\cdot w_i$ is a negative number, then $h_i\otimes w_i$ cannot be the Dehn invariant of a rectangle or of any polygon, because it would have negative area. For this reason, the rank of the Dehn invariant is a lower bound on the number of rectangles that can be obtained in an orthogonal dissection, but it requires an additional argument to prove that these two numbers are equal.

\section{Geometric realizability}

In the case of the polyhedral Dehn invariant, not every tensor in the space describing these invariants comes from the Dehn invariant of a polyhedron. There exists a surjective homomorphism of groups from the tensor space $\mathbb{R}\otimes_{\mathbb{Z}}\mathbb{R}/\mathbb{Z}$ onto the group $\Omega^1_{\mathbb{R}/\mathbb{Q}}$ of Kähler differentials, such that the tensors coming from Dehn invariants are exactly those mapped to the group identity. The preimages of nonzero Kähler differentials are tensors that do not come from Dehn invariants~\cite{Dup-NTM-01}. In contrast, for the orthogonal Dehn invariant, the only obstacle to geometric realizability is area:

\begin{figure}[t]
\centering\includegraphics[width=0.6\textwidth]{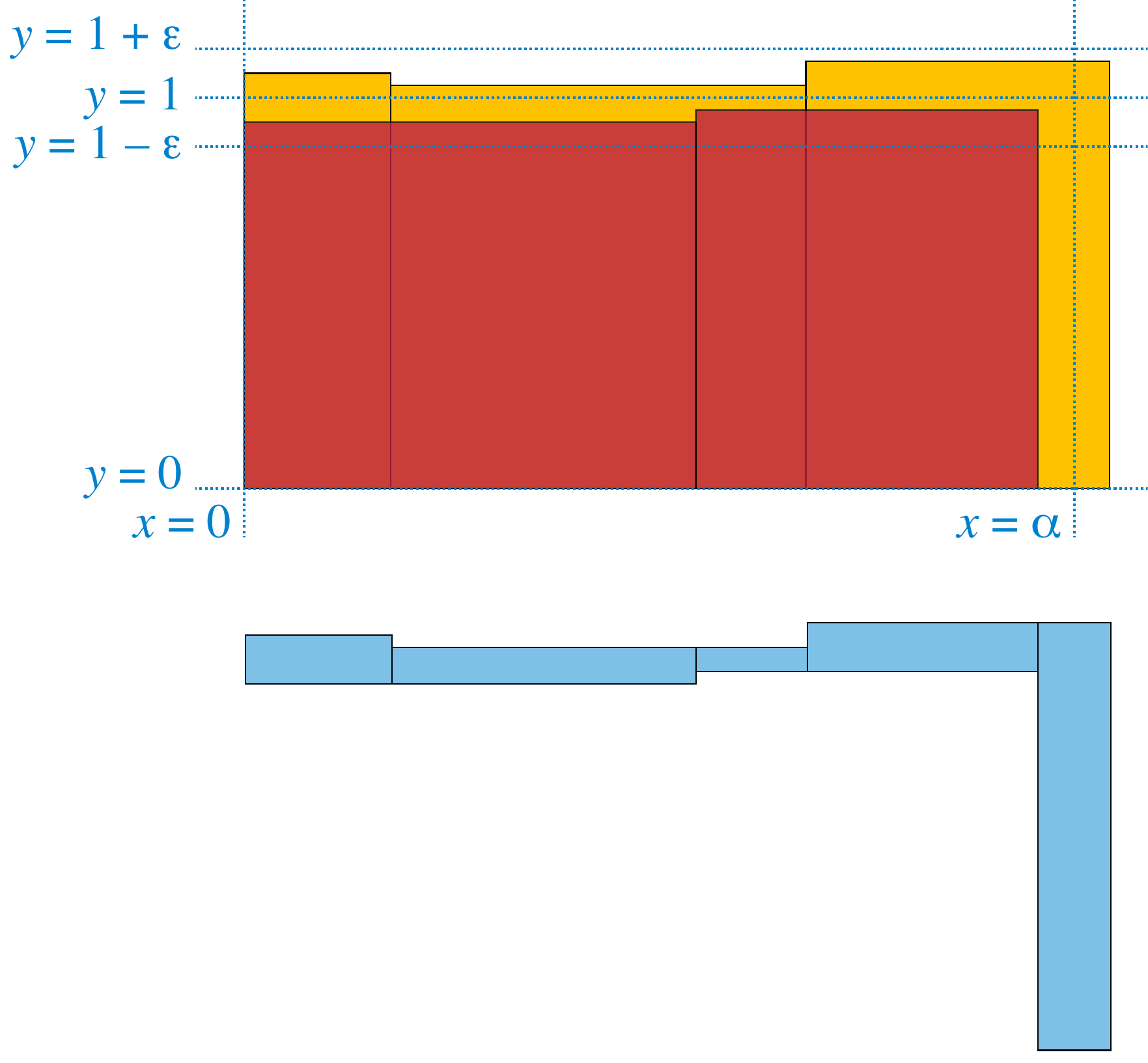}
\caption{Illustration for \cref{lem:realizability}: realizing each term in a tensor by a rectangle of height near one,
forming the difference of the positive and negative rectangles, and repartitioning the result into rectangles, produces a set of $r$ rectangles having a given Dehn invariant of rank $r$.}
\label{fig:realizability}
\end{figure}

\begin{lemma}
\label{lem:realizability}
Let $D=\sum_{i=1}^r h_i\otimes w_i$ be a tensor of rank $r$ in $\mathbb{R}\otimes_{\mathbb{Q}}\mathbb{R}$, and suppose that the putative area $a(D)=\sum_{i=1}^r h_i\cdot w_i$ is positive. Then $D$ is the Dehn invariant of a disjoint union of $r$ rectangles.
\end{lemma}

\begin{proof}
Partition the range of indices $[1,r]$ into two subsets $I^+$ and $I^-$, where $i\in I^+$ if $h_i\cdot w_i>0$ and in $i^-$ otherwise. (Because each term contributes to the rank, it is not possible for $h_i\cdot w_i$ to equal zero.)
Let $a^+=\sum_{i\in I^+} h_i\cdot w_i$ and $a^-=-\sum_{i\in I^-} h_i\cdot w_i$, so that $a(D)=a^+-a^-$. 
By assumption this is positive. We may assume without loss of generality that both sets of indices are non-empty: $I^+$ non-empty is needed to make $a(D)$ positive, and if $I^-$ is empty then we can represent $D$ using the disjoint union of rectangles of height $|h_i|$ and width $|w_i|$ without any additional construction. We can find two rational numbers $\alpha$ and $\varepsilon>0$ such that $a^-<
\alpha<a^+$, with $a^-<\alpha(1-\varepsilon)$ and $a^+>\alpha(1+\varepsilon)$.\footnote{Computationally, this uses the assumption from our model of computation that we can find a rational number between two products of combinations of basis elements.} These numbers are illustrated with the dashed blue axis-parallel lines in \cref{fig:realizability}.

Let $A$ be a rectangle with unit height and with width $\alpha$. For each index $i\in I^+$, find a positive rational number $q_i$ such that $1<q_i\cdot |h_i|<1+\varepsilon$, and construct a rectangle of height $q_i\cdot |h_i|$ and width $|w_i|/q_i$, whose Dehn invariant is $h_i\otimes w_i$. Arranging these rectangles side by side on a common baseline produces an orthogonal polygon $P^+$ whose height varies between $1$ and $1+\varepsilon$, whose area is $a^+$, and whose Dehn invariant is $\sum_{i\in I^+} h_i\otimes w_i$. In order to achieve area $a^+$ with height everywhere less than $1+\varepsilon$, $P^+$ must have width greater than $a^+/(1+\varepsilon)>\alpha$, so it completely covers $A$. These side-by-side rectangles are shown in yellow in the top part of \cref{fig:realizability}.

In the same way, for each index $i\in I^-$, find a positive rational number $q_i$ such that $1-\varepsilon<q_i\cdot |h_i|<1$,
and construct a rectangle of height $q_i\cdot |h_i|$ and width $|w_i|/q_i$, whose Dehn invariant is $-h_i\otimes w_i$.
Arranging these rectangles side by side on a common baseline produces an orthogonal polygon $P^-$ whose area is $a^-$ and whose Dehn invariant is $-\sum_{i\in I^-} h_i\otimes w_i$, entirely within $A$, the red rectangles  in the top part of \cref{fig:realizability}.

Arranging $P^+$ and $P^-$ so they share the same bottom left vertex, and computing the set-theoretic difference $P^+\setminus P^-$, produces a polygon $P$ whose Dehn invariant is $D$. It can be sliced vertically at each vertex whose $x$-coordinate is intermediate between its smallest and largest $x$-coordinate, as shown in the bottom part of \cref{fig:realizability}. There are $r-1$ slices (one for each side where two rectangles from $I^+$ meet, and one for each left side of a rectangle from $I^-$), so the result is a set of $r$ rectangles with total Dehn invariant $D$, as required.
\end{proof}

\section{Dissectability}

Long after the work of Dehn, Sydler proved that the polyhedral Dehn invariant is a complete invariant: any two polyhedra with the same volumes and Dehn invariants can be dissected to each other~\cite{Syd-CMH-65}. We need an analogous result for the orthogonal Dehn invariant. We do not bound the number of pieces in a dissection. It is not possible to bound this number of pieces by any function of the number of input vertices, because even the trivial dissection of a $1\times n$ rectangle into an $n\times 1$ rectangle requires $n$ pieces, a number that can be made arbitrarily large while keeping the number of input vertices constant.

\begin{theorem}
\label{thm:dissectability}
Any two orthogonal polygons with the same Dehn invariant have an orthogonal dissection.
\end{theorem}

\begin{figure}[t]
\centering
\includegraphics[width=0.8\textwidth]{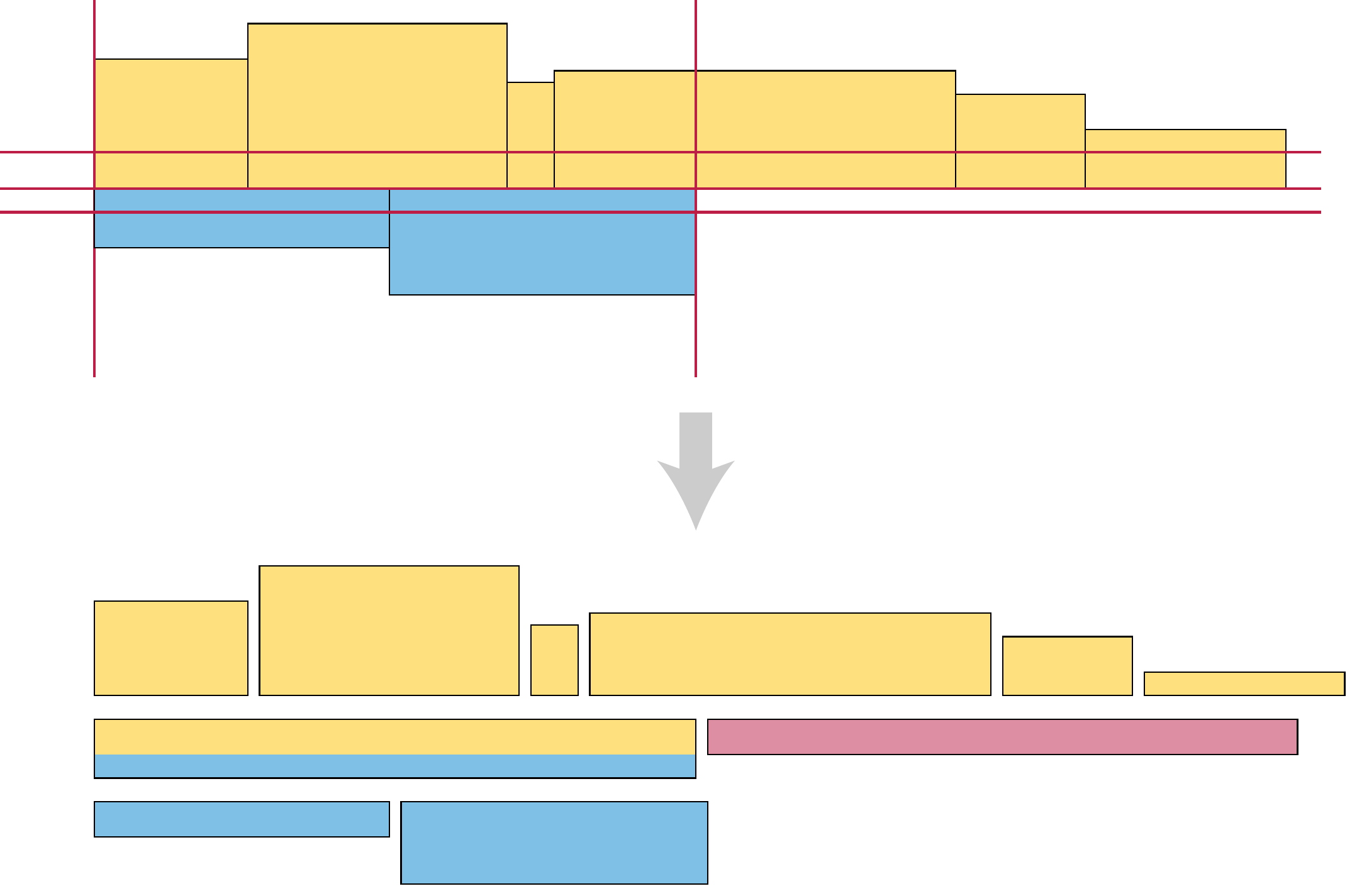}
\caption{Illustration for \cref{thm:dissectability}. The horizontal red lines are (from top to bottom) $y=\varepsilon^+$, $y=0$, and $y=-\varepsilon^-$; the vertical lines are (left to right) $x=0$ and $x=\min\{c_i^+,c_i^-\}$. Slicing the rectangles in $R_i^+$ (yellow) and $R_i^-$ (blue) by these lines dissects them into a family of rectangles whose heights do not depend on $\hat y$ (the bottom blue and yellow rectangles) together with a single rectangle whose coefficient of $\hat y$ is $\pm 1$ (red).}
\label{fig:dissectability}
\end{figure}

\begin{proof}
We may assume without loss of generality that the two polygons  $P_1$ and $P_2$ have already been dissected into (different) disjoint sets of rectangles $R_1$ and $R_2$. We use induction on the size of rational bases for the heights and widths of these rectangles (which may be a superset of a basis for the Dehn invariant). As a base case, if these bases have size one, all rectangles have heights and widths that are rational multiples of each other. In this case we can scale the $x$ and $y$ coordinates separately to clear denominators in these coordinates and make all rectangle side lengths integers, allowing a dissection using unit squares.

Otherwise, by the symmetry of heights and widths, we can assume without loss of generality that the $y$-basis has at least two elements; let $\hat y$ be one of them. For each rectangle in $R_1$ and $R_2$, of width $w_i$ and height $h_i$ let $q_i$ be the coefficient of $\hat y$ in the expansion of $h_i$ as a rational combination of basis elements. Whenever $q_i\ne 0$, apply the base case of the theorem (for the one-element bases $\{w_i\}$ and $\{h_i\}$) to dissect that rectangle into another rectangle of width $q_i\cdot w_i$ and height $h_i/q_i$. After this step, for all rectangles in $R_1$ and $R_2$, the coefficient of $\hat y$ in the rectangle height belongs to $\{-1,0,1\}$. Let $R_i^+$ be the rectangles in $R_i$ for which this coefficient is $1$, $R_i^-$ be the rectangles for which it is $-1$, and $R_i^0$ be the rectangles for which it is~$0$. By composition of dissections, a dissection of these modified sets of rectangles into each other will lead to a dissection of $P_1$ and $P_2$ into each other. 

For each of $P_1$ and $P_2$, translate the rectangles of $R_i^+$ so they are placed side by side, with their bottom sides all placed on the $x$-axis, with the left side of the leftmost rectangle placed on the $y$-axis. Similarly translate the rectangles of $R_i^-$ so they are side by side, with their top sides all placed along the $x$-axis, and again with the left side of the leftmost rectangle placed on the $y$-axis. Let $\varepsilon$ be the smallest height of any rectangle in either $R_1$ or $R_2$.
Choose two numbers $0<\varepsilon^+<\varepsilon$ and $0<\varepsilon^-<\varepsilon$, so that both of these numbers are expressible as a rational combination of elements of the $y$-basis, with the coefficient of $\hat y$ in $\varepsilon^+$ equal to $1$ and the coefficient of $\hat y$ in $\varepsilon^-$ equal to $-1$. (The ability to make this choice hinges on the fact that the rational multiples of any remaining basis element are dense in the real number line.) Let $c_i^+$ be the $x$-coordinate of the right end of the rightmost rectangle in the placement of $R_i^+$, and define $c_i^-$ symmetrically.

Now that the rectangles have been placed in this way, slice them by the horizontal lines $y=\varepsilon^+$ and $y=-\varepsilon^-$. This leaves a hexagonal region between these two lines, which we dissect into two rectangles by slicing it with the vertical line $x=\min\{c_i^+,c_i^-\}$ (two different lines, one for $R_1$ and the other for $R_2$). The dissection is shown in \cref{fig:dissectability}.

The rectangles in $R_i^0$, the remaining parts of rectangles in $R_i^+$ above the line $x=\varepsilon^+$, and the remaining parts of rectangles in $R_i^-$ below the line $x=\varepsilon^-$, all have heights whose rational expansion in terms of the $y$-basis does not use $\hat y$. The rectangle to the left of the vertical slice line, and between the two horizontal slice lines, has height $\varepsilon^++\varepsilon^-$; here, the coefficients of $\hat y$ cancel leaving a rectangle height whose expansion in terms of the basis does not use $\hat y$. This leaves all dependence on $\hat y$ concentrated in one remaining rectangle, to the left of the vertical slice line, with height $\varepsilon^+$ or $\varepsilon^-$ and width $|c_i^+-c_i^-|$. Let $\hat w_i$ denote this width.

Because all remaining pieces except this rectangle have heights that do not depend on $\hat y$, it follows that the coefficients of $\dehn(P_i)$, in the row of the coefficient matrix corresponding to basis element $\hat y$,
are exactly the coefficients in the rational expansion of $\hat w_i$ for a rectangle of height $\varepsilon^+$, or the negation of those coefficients for a rectangle of height $\varepsilon^-$. By the assumption that $\dehn(P_1)=\dehn(P_2)$, these matrix coefficients must be equal. The widths $\hat w_1$ and $\hat w_2$ of the rectangles can be recovered, up to their signs, as the number represented by these coefficients using the $y$-basis. It is not possible for the two signs to be different, because both $\hat w_1$ and $\hat w_2$ are non-negative. Therefore, the two remaining rectangles must both have the same height, $\varepsilon^+$ or $\varepsilon^-$, and the same width, $\hat w_1=\hat w_2$, and need no more dissection to be transformed into each other.

We have shown that $P_1$ and $P_2$ can be dissected into two congruent rectangles whose height expansion uses $\hat y$, and into a larger number of additional rectangles whose height expansion does not use $\hat y$. These remaining rectangles have a smaller basis for their heights and (because we have removed a congruent rectangle from each polygon) have equal Dehn invariants. The result follows from the induction hypothesis.
\end{proof}

\section{Putting the pieces together}

We are now ready to prove our main results:

\begin{theorem}
\label{thm:minrect-rank}
The minimum number of rectangles into which an orthogonal polygon can be dissected by axis-parallel cuts and translation equals the rank of its orthogonal Dehn invariant.
\end{theorem}

\begin{proof}
This number of rectangles is lower-bounded by the rank, by the discussion in \cref{sec:rank}. If the rank is $r$, then there exists a set of $r$ rectangles with the same invariant as the polygon, by \cref{lem:realizability}.
The given polygon can be dissected into these rectangles, by \cref{thm:dissectability}.
\end{proof}

\begin{theorem}
\label{thm:minrect-algorithm}
We can compute the minimum number of rectangles into which an orthogonal polygon can be dissected, given a representation for its coordinates over a rational basis, in a polynomial number of rational-arithmetic operations.
We can construct a minimal set of rectangles into which it can be dissected, in a polynomial number of operations using arithmetic over the given rational basis.
\end{theorem}

\begin{proof}
To compute the rank, we compute the Dehn invariant as described in \cref{sec:matrix}, and apply any polynomial-time algorithm for computing the rank of a rational-number matrix, such as Gaussian elimination. To construct the rectangles, we follow the steps in the proof of \cref{lem:realizability}, which uses only a polynomial number of operations involving comparing linear combinations of basis elements and finding rational approximations to them.
\end{proof}

We remark that, as well as counting rectangles, the rank of the orthogonal Dehn invariant can also count edges:

\begin{observation}
\label{obs:edges}
Let $r$ be the minimum number of rectangles into which a given orthogonal polygon has an orthogonal dissection and let $s$ be the minimum number of edges of a polygon into which it has an orthogonal dissection. Then $s=2r+2$, and there exists a simple polygon with $s$ edges into which it has an orthogonal dissection.
\end{observation}

\begin{proof}
In one direction, suppose that a given orthogonal polygon has an orthogonal dissection into $r$ disjoint rectangles.
Line up these rectangles with a common baseline, side by side, and glue them together, producing a simple polygon with at most $2r+2$ edges. Therefore, $s\le 2r+2$, and a simple polygon with $\le 2r+2$ edges can be produced by an orthogonal dissection.

In the other direction, suppose that we can dissect the given orthogonal polygon into an orthogonal polygon with $s$ edges, by an orthogonal dissection. Because edges of orthogonal polygons alternate between horizontal and vertical, $s$ must be even, with $s/2$ horizontal edges. Each $x$-coordinate of a horizontal edge is shared by the next horizontal edge along the boundary of the polygon, so there must also be at most $s/2$ distinct $x$-coordinates. Among these, at most $s/2-2$ are non-extreme (neither the minimum nor the maximum coordinate value). Cutting the polygon by vertical lines through these non-extreme coordinates, and gluing together the rectangular pieces obtained within each resulting vertical slab between two such lines, produces a dissection into $s/2-1$ rectangles, so $r\le s/2-1$. Equivalently, $2r+2\le s$.
\end{proof}

\section{Dissection into prototiles}
\label{sec:prototile}

Another use of the polyhedral Dehn invariant, besides dissection of one shape into another, involves tiling. Any polyhedron that tiles space must have Dehn invariant zero, and any polyhedron with Dehn invariant zero can be dissected into a different polyhedron that tiles space. For the axis-parallel polygonal Dehn invariant we study, things don’t work out quite so neatly. The Greek cross can tile, but has nonzero Dehn invariant. More, any axis-parallel polygon can be cut into multiple rectangles, and these can tile space (non-periodically) by grouping them into rows of the same type of rectangle (\cref{fig:trirect-row-tiling}). So the Dehn invariant cannot be used to prove that such a thing is impossible, because it is always possible. If we could rotate pieces, we could also rearrange certain sets of more than two rectangles, such as the three rectangles of \cref{fig:trirect}, into a single-piece axis-parallel hexagon that could tile the plane periodically (\cref{fig:trirect-flip-tiling}).

\begin{figure}[t]
\centering\includegraphics[width=0.6\textwidth]{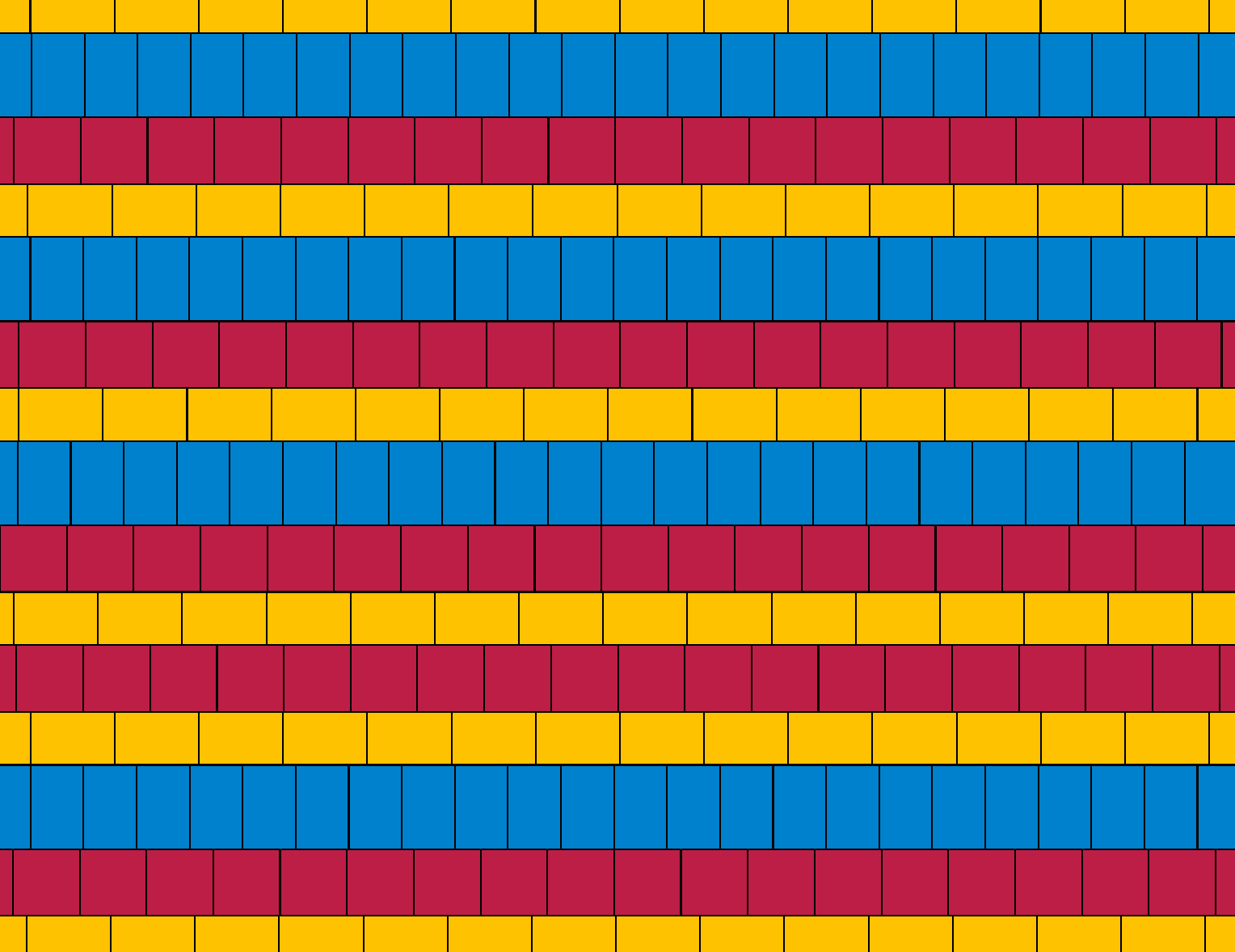}
\caption{Construction of a non-periodic tiling from an arbitrary dissection into rectangles}
\label{fig:trirect-row-tiling}
\end{figure}

However, for the orthogonal dissections considered here, without rotation, the rank of the Dehn invariant does produce a limitation on the ability to tile periodically without rotation. Here, we follow Grünbaum and Shephard~\cite{GruShe-89} in defining a \emph{periodic} tiling to be a tiling that has a two-dimensional lattice of translational symmetries. For some other authors, this would be called a 2-periodic tiling, to distinguish it from tilings that have one-dimensional but not two-dimensional translational symmetry; we do not make this distinction.

\begin{theorem}
\label{thm:rank-vs-tiling}
An orthogonal polygon $P$, or any finite number of copies of $P$, has an orthogonal dissection to a prototile that can tile the plane periodically if and only if the rank of its Dehn invariant is at most two.
\end{theorem}

\begin{proof}
If the rank of the Dehn invariant is one, $P$ can be dissected to a rectangle, which tiles periodically. If the rank is two, $P$ can be dissected into two rectangles, and reassembled into a hexagon, which (like the prototiles of \cref{fig:trirect-flip-tiling}) tile periodically.

\begin{figure}[t]
\centering\includegraphics[width=0.6\textwidth]{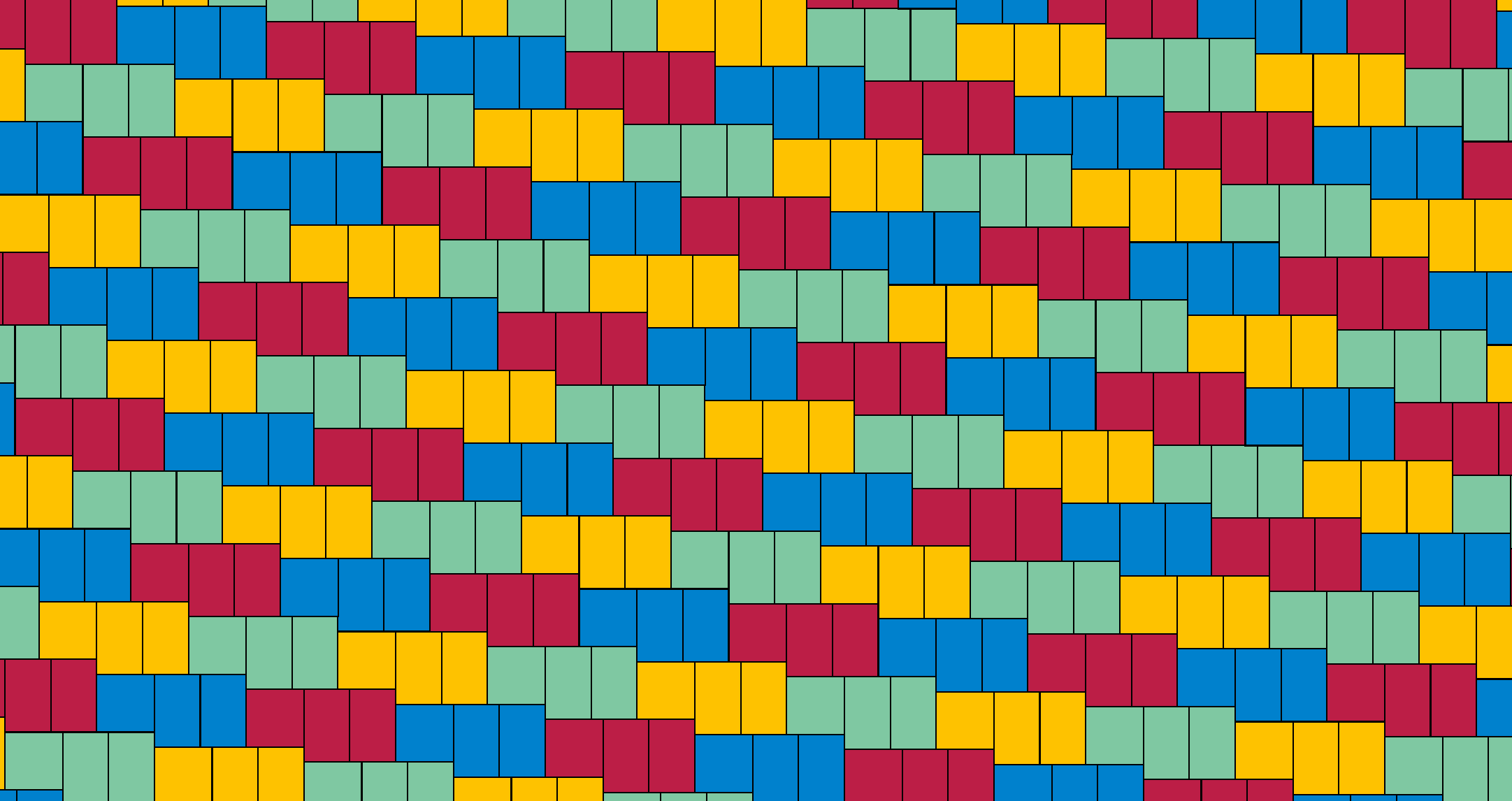}
\caption{With orthogonal rotations of pieces, the polygon of \cref{fig:trirect} can be dissected to form the prototile of a periodic tiling of the plane.}
\label{fig:trirect-flip-tiling}
\end{figure}

Combining $n$ copies of $P$ multiplies the Dehn invariant by the scalar $n$, which does not change the rank.  Every periodic tiling of the plane has a fundamental region in the shape of an axis-parallel hexagon, like the prototiles of \cref{fig:trirect-flip-tiling}. (Because it tiles by translation, this fundamental region may combine a finite number of prototiles of the tiling.) If copies of $P$ could be dissected to the prototile of a tiling, they could also be dissected to this fundamental region, which has Dehn invariant at most two.
\end{proof}

In particular, as a shape whose Dehn invariant has rank three, the orthogonal polygon of \cref{fig:trirect} has no orthogonal dissection to a prototile for a periodic tiling of the plane.

\section{Orthogonal dissections with rotation}

In this section we extend the notion of an orthogonal dissection to allow $90^\circ$ rotations, or equivalently reflections across a line with slope $\pm 1$. These two notions are equivalent because, for any orthogonal dissection, we can subdivide the pieces of the dissection into rectangles, for which these rotations and reflections have the same effect.

\begin{definition}
If $P$ is any orthogonal polygon, let $P\transpose$ denote the reflection of $P$ across a line of slope $-1$. We call $P\transpose$ the \emph{transpose} of $P$ by analogy to the transpose of a matrix or tensor, for which we use the same notation.
\end{definition}

Transposition (of polygons or matrices) commutes with taking the Dehn invariant:

\begin{observation}
For all orthogonal polygons $P$, and for any representation of the Dehn invariant $\mathcal D$ as a matrix over a combined basis,
$\dehn(P\transpose)=\dehn(P)\transpose$.
\end{observation}

\begin{proof}
This follows by subdividing $P$ into rectangles and observing that the effect of transposition on a rectangle commutes with the direct formula for the Dehn invariant of the rectangle.
\end{proof}

Transposition of Dehn invariants can also equivalently be interpreted in a coordinate-free way, as transposition of tensors, a bilinear operation that maps $x\otimes y$ to $y\otimes x$ for all $x$ and $y$. 
We use transposition to define a symmetric form of the Dehn invariant, its \emph{symmetric part}:

\begin{definition}
For an orthogonal polygon $P$, define its symmetrized Dehn invariant as the average of the Dehn invariant and its transpose:
\[\symdehn(P)=\frac12\bigl(\dehn(P)+\dehn(P)\transpose\bigr).\]
\end{definition}

\begin{observation}
The symmetrized Dehn invariant is an invariant of orthogonal dissection with rotation.
\end{observation}

\begin{proof}
Suppose that $P$ and $P'$ are any two polygons related by an orthogonal dissection with rotation. We can assume without loss of generality, by adding additional subdivisions if necessary, that the dissection consists of subdividing $P$ into rectangles, rotating a subset of these rectangles, and translating them so that they form a subdivision of $P'$. The subdivision and translation steps, and the step in which the translated pieces are glued together to form $P'$, change neither $\dehn$ nor $\dehn\transpose$.
The step of rotating any rectangle acts on the symmetrized Dehn invariant by swapping $\dehn$ with $\dehn\transpose$, leaving their average unchanged.
\end{proof}

\begin{lemma}
\label{lem:symmetrize}
Any orthogonal polygon $P$ can be dissected with rotation into a polygon $\symmetrize P$ for which $\symdehn(P)=\dehn(\symmetrize{P})$.
\end{lemma}

\begin{proof}
Subdivide $P$ into rectangles, cut each rectangle by an axis-parallel line (either horizontal or vertical) into two congruent rectangles, and rotate one of these two rectangles. Glue the results together arbitrarily to form $\symmetrize P$.
\end{proof}

\begin{corollary}
Any two orthogonal polygons with the same symmetrized Dehn invariants have a dissection with rotation into each other.
\end{corollary}

\begin{proof}
Let the symmetrized invariant of both polygons be $\symmetrize{\dehn}$.
Dissect each of the given polygons into a polygon whose Dehn invariant is $\symmetrize{\dehn}$, by \cref{lem:symmetrize}, and apply \cref{thm:dissectability}.
\end{proof}

We note that the effect of symmetrization on rank is not uniform. For instance, if $P$ is a rectangle with incommensurable sides, its Dehn invariant and symmetrized Dehn invariant can be written in matrix form (using the sides as a basis) as
\[
\dehn(P)=\begin{pmatrix} 0 & 0 \\ 1 & 0 \\ \end{pmatrix},\qquad
\symdehn(P)=\begin{pmatrix} 0 & \tfrac12 \\ \tfrac12 & 0\\ \end{pmatrix},
\]
and in this case symmetrization doubles the rank. This is the most possible:

\begin{observation}
\label{obs:sym-2x-rank}
For any orthogonal polygon $P$, $\rank\bigl(\symdehn(P)\bigr)\le2\rank\bigl(\dehn(P)\bigr)$.
\end{observation}

\begin{proof}
This follows immediately from the definition of $\symdehn(P)$ as the sum of two tensors of equal rank, $\tfrac12\dehn(P)$ and its transpose.
\end{proof}

\begin{figure}[t]
\centering\includegraphics[width=0.6\textwidth]{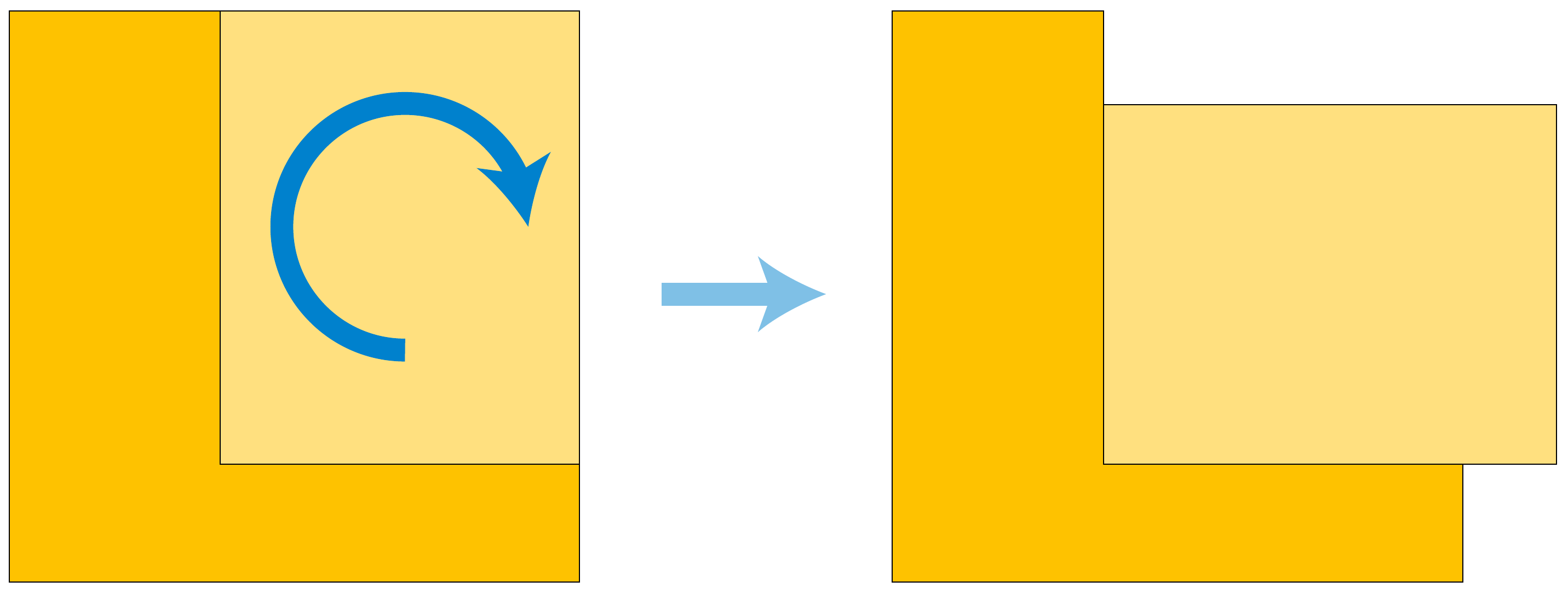}
\caption{Construction of a polygon $P$ (right) for which $\dehn(P)=3\symdehn(P)$.}
\label{fig:sym-reduces-rank}
\end{figure}

On the other hand it is easy to construct a polygon $P$ (\cref{fig:sym-reduces-rank}) whose Dehn invariant and symmetrized Dehn invariant (over an appropriate basis) are
\[
\dehn(P)=\begin{pmatrix} 1 & 0 & 0 \\ 0 & 0 & 1 \\ 0 & -1 & 0 \\ \end{pmatrix},\qquad
\symdehn(P)=\begin{pmatrix} 1 & 0 & 0 \\ 0 & 0 & 0 \\ 0 & 0 & 0 \\ \end{pmatrix}
\]
and in this case symmetrization decreases the rank from three to one. A similar construction can produce examples in which the rank decreases from arbitrarily large values to one. Therefore, the minimum number of rectangles in an orthogonal dissection of a polygon without rotation does not accurately approximate the minimum number of rectangles in a dissection with rotation. 

Returning to finding dissections into a minimum number of rectangles, we have the following formula:

\begin{theorem}
\label{thm:minrect-rot}
Let $P$ be any orthogonal polygon, and let $r$ denote the minimum number of rectangles that can be formed from $P$ by an orthogonal dissection with rotation. Then
\[
r = \min\left\{
\rank(X)\Bigm\vert
X\in\mathbb{R}\otimes_{\mathbb{Q}}\mathbb{R} \textnormal{ and }
 \frac12(X+X\transpose)=\symdehn(P)\right\}.
\]
\end{theorem}

\begin{proof}
Let $R$ be a set of $r$ rectangles obtained from $P$ by orthogonal dissection with rotation, and let $X=\dehn(R)$. Clearly, $r=\rank(X)$, and
\[
\frac12(X+X\transpose)=\symdehn(R)=\symdehn(P).
\]
There can be no other tensor $Y$ with smaller rank and with
\[
\frac12(Y+Y\transpose)=\symdehn(P),
\]
for if there were we could find a realization of $Y$ as a polygon, dissect $P$ into that realization, and then dissect the realization into fewer rectangles.
\end{proof}

Unfortunately we do not know how to efficiently find a minimizing tensor $X$ and compute this minimum number of rectangles. It is a special case of the $\mathsf{NP}$-hard affine rank minimization problem, in which one wishes to find a minimum-rank matrix in a linear subspace of matrices~\cite{JaiMekDhi-NeurIPS-10}, but over rational rather than the more usual real matrices. Instead we have the following approximation algorithm.

\begin{theorem}
We can approximate the minimum number of rectangles into which an orthogonal polygon can be dissected with rotation, in polynomial time, to within an approximation ratio of 2.
\end{theorem}

\begin{proof}
Compute and return $\rank\bigl(\symdehn(P)\bigr)$. Because $\symdehn(P)$ is one of the choices for $X$ in \cref{thm:minrect-rot}, the result is at least equal to the minimum number of rectangles.
For the optimal $X$ of \cref{thm:minrect-rot}, the result we return is within a factor of two of the value obtained from $X$, by \cref{obs:sym-2x-rank}.
\end{proof}

\section{Conclusions}

We have shown that the rank of the orthogonal Dehn invariant of an orthogonal polygon controls the number of rectangles into which it can be dissected by axis-parallel slices and translation, leading to a polynomial time algorithm to compute this number of rectangles or to construct an optimal set of rectangles into which it can be dissected. The dissection itself may require a non-polynomial number of pieces. The number of rectangles, in turn, controls the ability to dissect a polygon into a shape that tiles the plane. Answering a question posed in the conference version of this paper, we also find an approximate extension of these results to dissections that allow $90^\circ$ rotations.

The rank of the polyhedral Dehn invariant, similarly, provides a lower bound on the number of edges of a polyhedron into which a given polyhedron may be dissected, because every polyhedron's Dehn invariant is defined as a sum of tensors over its edges, with rank at most the number of edges in the sum. Because a tetrahedron has six edges, the rank of the polyhedral Dehn invariant, divided by six, also gives a lower bound on the number of tetrahedra into which a given polyhedron may be dissected.
It is tempting to guess that, rather than merely lower-bounding these numbers, the Dehn invariant is a constant-factor approximation both to the minimum number of edges in a single polyhedron resulting from a dissection and to the minimum number of tetrahedra in a dissection into disjoint tetrahedra. However, we have been unable to prove this. What would be needed is a construction of a polyhedron with a given Dehn invariant and with a number of edges proportional to the rank of the invariant, analogous to \cref{lem:realizability}, but this is made more difficult by the fact that not all tensors are realizable as polyhedral Dehn invariants. We do not even have a proof that the minimum number of edges and the minimum number of tetrahedra are within constant factors of each other; there exist polyhedra that cannot be subdivided into a linear number of tetrahedra (or more generally a linear number of convex pieces) relative to their numbers of edges~\cite{PatYao-DCG-90}, but this does not rule out more parsimonious dissections for such examples.

Another natural direction for future research concerns dissections into squares rather than rectangles. Squares have Dehn invariants that are symmetric, of rank one, with positive area. Therefore, for a dissection (without rotation) into squares to exist, the Dehn invariant must be symmetric. Any symmetric Dehn invariant can be decomposed into a sum of rank-one symmetric Dehn invariants, but these might not have positive area. We do not know how to determine when a decomposition into rank-one symmetric positive-area Dehn invariants exists, nor how to minimize the number of terms in such a decomposition.

Order-two tensors have additional invariants beyond their rank, such as those obtained as the coefficients of the characteristic polynomial or as any function of those coefficients. (The rank can be obtained in this way from the difference in degrees of the highest-degree and lowest-degree nonzero coefficients.) Another example is the minimum rank of the tensors with the same symmetric part as a given tensor, as used here for minimizing rectangles under dissection with rotation. Our work naturally raises the questions: which other invariants are meaningful for dissection problems, what do they mean, and how efficiently can they be computed?

\section*{Acknowledgements}

This research was supported in part by NSF grant CCF-2212129. A preliminary announcement of some of these results appeared in the proceedings of the 34th Canadian Conference on Computational Geometry, 2022, pp. 143–150.

\bibliographystyle{plainurl}
\bibliography{dehn}

\end{document}